\documentclass{amsart}

\usepackage{amsmath}
\usepackage{amsthm}
\usepackage{amssymb}
\usepackage{amsthm}
\usepackage{graphicx}
\usepackage{caption}
\usepackage{subcaption}
\usepackage[left=1.4in, right=1.4in, top=1.5in, bottom=1.5in]{geometry}

\allowdisplaybreaks[3]

 \newtheorem{thm}{Theorem}[section]

 \theoremstyle{definition}
 
 \theoremstyle{remark}
 
 \numberwithin{equation}{section}

 \newcommand{\h}{\mathcal{H}}

 \newcommand{\abs}[1]{\left\vert#1\right\vert}

\author{Johan G.B. Beumee$^\dag$ \\Chris Cormack$^\dag$ \\Manish Patel$^\dag$ \\Peyman Khorsand$^\ddag$}
\thanks{$^\dag$MPCapital Advisory Services LLP}
\thanks{$^\ddag$Jefferies Bank}

\begin{document}

\email{johan.beumeeATbtinternet.com, jbeumeeATmpcapital.co.uk}

\thanks{This work was completed with support from MP Capital.}

\date{}



\begin{abstract}
This paper investigates the position (state)
distribution of the single step binomial
(multi-nomial) process on a discrete state / time grid
under the assumption that the velocity process rather
than the state process is Markovian. In this model the
particle follows a simple multi-step process in
velocity space which also preserves the proper state
equation of motion.  Many numerical numerical examples
of this process are provided.  For a smaller grid the
probability construction converges into a correlated
set of probabilities of hyperbolic functions for each
velocity at each state point. It is shown that the two
dimensional process can be transformed into a
Telegraph equation and via transformation into a
Klein-Gordon equation if the transition rates are
constant.  In the last Section there is an example of
multi-dimensional hyperbolic partial differential
equation whose numerical average satisfies Newton's
equation. There is also a momentum measure provided
both for the two-dimensional case as for the
multi-dimensional rate matrix.
\end{abstract}

\title{Path Diffusion, part I}
\maketitle

\section*{Introduction}
This paper investigates the position (state)
distribution of the single step binomial process
assuming that the velocity on the node grid rather
than the state process is Markovian. Under this
assumption the particle steps up or down in velocity
following a simple (or multi)-step process on a fixed
grid with discrete time preserving the proper equation
of motion. The equations of motion are defined as a
joint probability per velocity and state as a function
of time and state with a state transition matrix, see
Section \ref{Construction1}.  The velocity rate matrix
is the consequence of the Markov assumption.

The two-factor (multi-factor) solution to the binomial
probability Markovian process has different types of
solutions, see the numerical examples in Section 2.
For very small rate transitions the final probability
distribution may show the original velocity
information and transport the original conditions into
the future leaving a small amount of residuals. Or the
rate probabilities are considerable strongly
re-bunching the distribution into something that looks
like a Gaussian process. The final distribution may
therefore have individual peaks reflecting velocity
distributions or if the rates are large enough it
shows a central single modal distribution which
listens to a mean.

For a much smaller grid and constant rates the
probability equations converge into a correlated set
of probabilities of hyperbolic functions for each
velocity in state point.  The two dimensional case can
be transformed into a Telegraph equation for the state
density which can be transformed into a Klein-Gordon
equation if the transition rates are constant.  An
average velocity from the state can be defined as well
as Section \ref{Continuous1} and Section
\ref{MultiDimensional1} show.

This equation can be done in two velocity spaces or in
an infinite number is a set of diffusion equations of
them.  Both for the two-dimensional applications and
the multi-dimensional case a forward and a backward
velocity can be found as Section III and Section V
show. In the last Section there is multi-dimensional
hyperbolic partial differential equation whose average
satisfies Newton's equation.

\section{The Construction}
\label{Construction1}
\subsection*{Equation}
To construct the process consider the discrete space,
discrete time grid as defined in Figure
 ~\ref{NodeTransitionOrigv1}. The process is given by
the nodes
\begin{align*}
x(t) &= m\Delta x, m=...,-1, 0,
1, 2, ...
\\
t &=0,h,2h,...
\end{align*}
where $\Delta x$ is a small space grid size (vertical)
and $h$ is a discrete grid for the time (horizontal).
Figure  ~\ref{NodeTransitionOrigv1} shows a process
representing a particle stepping across an infinite
node grid at discrete time intervals $0,h,2h,...$. At
every time the process steps up or steps down only one
node from any given node in the grid. Clearly the
process can step up and down many nodes at a time but
for the moment consider only the simplest case.

\begin{figure}
\centering
\includegraphics[width=4in]{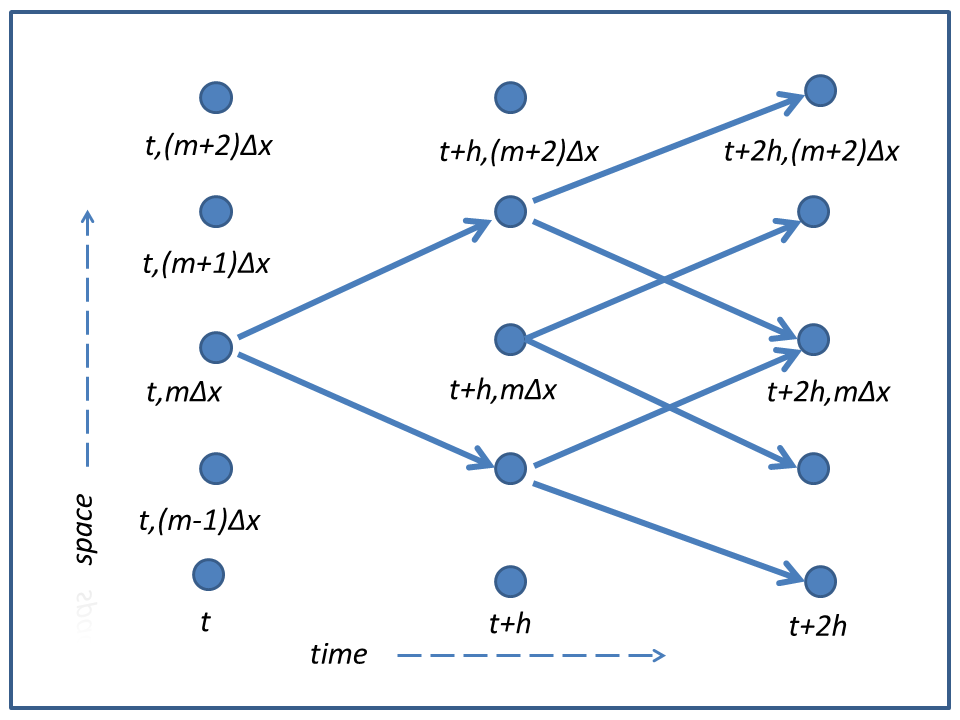}
\caption{Typical Mid-Grid Node Transitions}
\label{NodeTransitionOrigv1}
\end{figure}

If this process is Markovian in the state $x$ the
probability of being in state $x(t+h)$ depends on the
state $x(t)$ but not on any state before $t$.  By
design then a probability can be constructed on
$x(t+h)=m\Delta x,m=...-1,0,1,...$ from probabilities
in $x(t)$.  A Gaussian or backward equation can be
constructed in this fashion by considering the limit
for the outcome space by scaling $\Delta x \backsim
\sqrt(h)$ and then calculating the final distribution
from increasing evolutions ~\cite{FELLER}.

However, in this paper we assume that the process is
Markovian in the process velocity $(x(t),x(t+h))$
rather than the process state position $x(t)$.  This
means that any transition $(x(t),x(t+h))$ depends on
all previous transitions $(x(t-h),x(t))$ but not on
any previous transitions.

\begin{figure}
\centering
\includegraphics[width=4in]{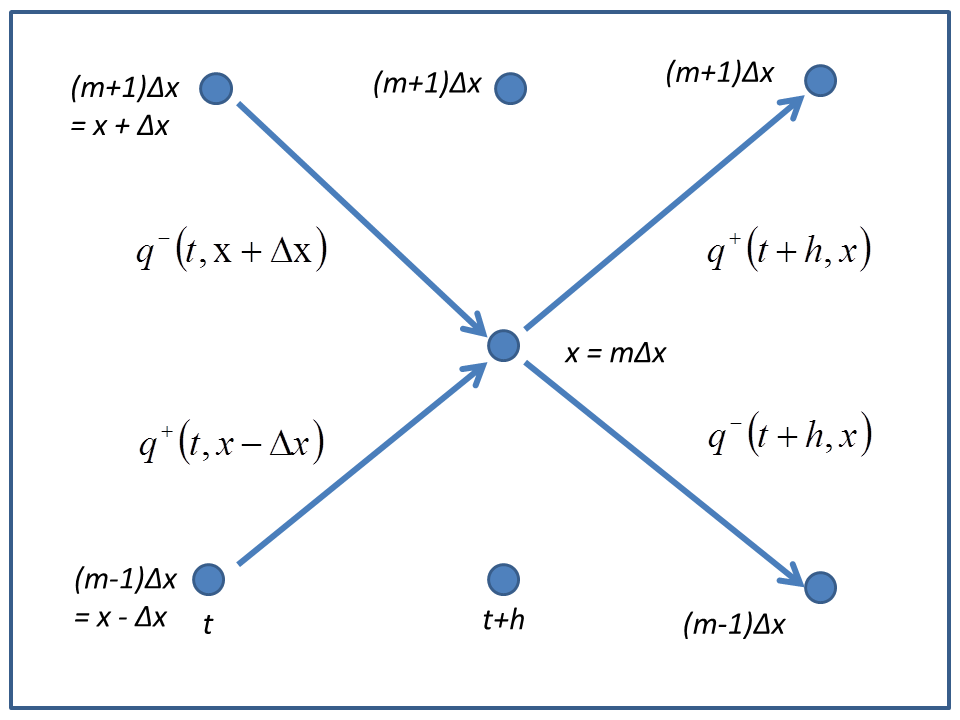}
\caption{Node Transition Analysis}
\label{NodeTransitionGraph}
\end{figure}

An example of this process is shown in Figure
~\ref{NodeTransitionGraph}.  If the particle resides
at time $t+h$ in the node in the middle of the grid
then $x(t+h)=m\Delta x = x$ for some $m$. Since it can
only step up or down from this state the only possible
outcomes at time point $t+2h$ are
\begin{align*}
x(t+2h)&=x+\Delta=(m+1)\Delta x
\\
&\text{or}
\\
x(t+2h)&=x-\Delta=(m-1)\Delta x.
\end{align*}
Similarly, to be in grid node $m\Delta x$ at time
$t+h$ the particle must have stepped down / up from
\begin{align*}
x(t)&=x+\Delta=(m+1)\Delta x
\\
&\text{or}
\\
x(t)&=x-\Delta=(m-1)\Delta x.
\end{align*}
The grid and the joint transition probabilities for
the stepping process are shown in Figure
~\ref{NodeTransitionGraph} for the gridpoints around
$t+h,m\Delta x$.

Now let $x=m\Delta x$ in Figure
~\ref{NodeTransitionGraph} and let
\begin{align}
\label{INITIALTRANS1}
\begin{split}
q^+(t+h,x)=P[x(t+2h)=x+\Delta x,x(t+h)=x]
\\
q^-(t+h,x)=P[x(t+2h)=x-\Delta x,x(t+h)=x]
\end{split}
\end{align}
be the up/down joint probabilities of an up or down
step after the particle travels through $x$.  As the
velocity process is Markovian the joint process
\eqref{INITIALTRANS1} must be dependent only on
\begin{align*}
q^+(t,x-\Delta x)&=P[x(t+h)=x,x(t)=x-\Delta x]
\\
q^-(t,x+\Delta x)&=P[x(t+h)=x,x(t)=x+\Delta x]
\end{align*}
linearly or otherwise. Notice that
$q^+(t+h,x),q^-(t+h,x),q^+(t,x-\Delta
x),q^-(t+h,x+\Delta x)$ are all joint distributions.

Hence the Markovian assumption requires that
\begin{align}
\label{TRANSITION1}
\begin{pmatrix}
q^+(t+h,x) \\
q^-(t+h,x)
\end{pmatrix}=
\begin{pmatrix}
1-\alpha(t,x) & \beta(t,x)\\
\alpha(t,x) & 1-\beta(t,x)
\end{pmatrix}
\begin{pmatrix}
q^+(t,x-\Delta x) \\
q^-(t,x+\Delta x)
\end{pmatrix}
\end{align}
for specific constants $\alpha(t,x)$, $\beta(t,x)$.
Notice that the $\alpha(t,x)$, $\beta(t,x)$ parameters
do not have to be equal but must be positive and that
the columns must add up to one to conserve
probability. After the particle arrives in $x$ it can
only step up or step down.

Probabilistically the $\alpha(t,x)$ parameters
consider the probability that the particle travels
downward from $x$ to $x-\Delta x$ at time $t$ after
travelling up from $x-\Delta x$ to $x$ at time $t-h$.
Similarly the $\beta(t,x)$ parameters consider the
probability that the particle travels upward from $x$
to $x+\Delta x$ at time $t$ after travelling down from
$x+\Delta x$ to $x$ at time $t-h$.  Specifically
\begin{align}
\label{TRANSITION1b}
\begin{split}
\begin{matrix}
\alpha(t,x)=P_t[x-\Delta,x,x-\Delta x] & \text{=  probability being in
$x(t)=x$ and stepping}
\\
{} &{\text{down after traveling from $x(t-h)=x-\Delta$}}
\\
\beta(t,x)=P_t[x+\Delta,x,x+\Delta x] &{\text {=  probability being in
$x(t)=x$ and stepping}}
\\
{} &{\text{up after traveling from $x(t-h)=x+\Delta$}}.
\end{matrix}
\end{split}
\end{align}
Notice that as the timestep $h$ becomes smaller these
curvature probabilities become smaller as well.

The probability densities $q^+(t,x)$, $q^-(t,x)$ are
joint particle distributions of being in position
$x(t)=x$ and moving in the "up" or "down" direction at
the same time. So in fact the probability density
$\rho(t,x)$ of the state can be defined as
\begin{align*}
P[x(t)=x]=\rho(t,x)=q^+(t,x)+q^-(t,x)
\end{align*}
which provides the probability that the particle is in
state $x$ at time $t$.  Now a summation over all $x =
m\Delta x$ (summing over $m$) will add up to one.

\subsection*{Results}
Adding the two equations in \eqref{TRANSITION1} yields
\begin{align*}
\rho(t+h,x) = q^+(t+h,x)+q^-(t+h,x)= q^+(t,x-\Delta
x)+q^-(t,x+\Delta x)
\end{align*}
so that the probability of a particle being in $t+h,x$
is accumulated from the probability of the particle
being in $x+\Delta$ at time $t$ stepping down and the
probability of the particle being in $x-\Delta$
stepping up.  This condition conserves probability
specifically for the case where only up or down steps
are allowed and clearly is dictated by the particle's
motion.

This equation also implies that
\begin{align*}
\sum_{m}\rho(t+h,m\Delta x) &= \sum_{m} q^+(t,(m-1)\Delta x)+\sum_{m}
q^-(t,(m+1)\Delta x)
\\
&=\sum_{m} q^+(t,m\Delta x)+\sum_{m}
q^-(t,m\Delta x)
\\
&=\sum_{m} \rho(t,m\Delta x)=1
\end{align*}
which shows that state probability is conserved.  Also then
\begin{align}
\label{TRANSITION43}
\begin{split}
&\sum_{m} q^+(t,m\Delta)<1
\\
&\sum_{m} q^-(t,m\Delta)<1.
\end{split}
\end{align}

\begin{figure}
\centering
\includegraphics[width=4in]{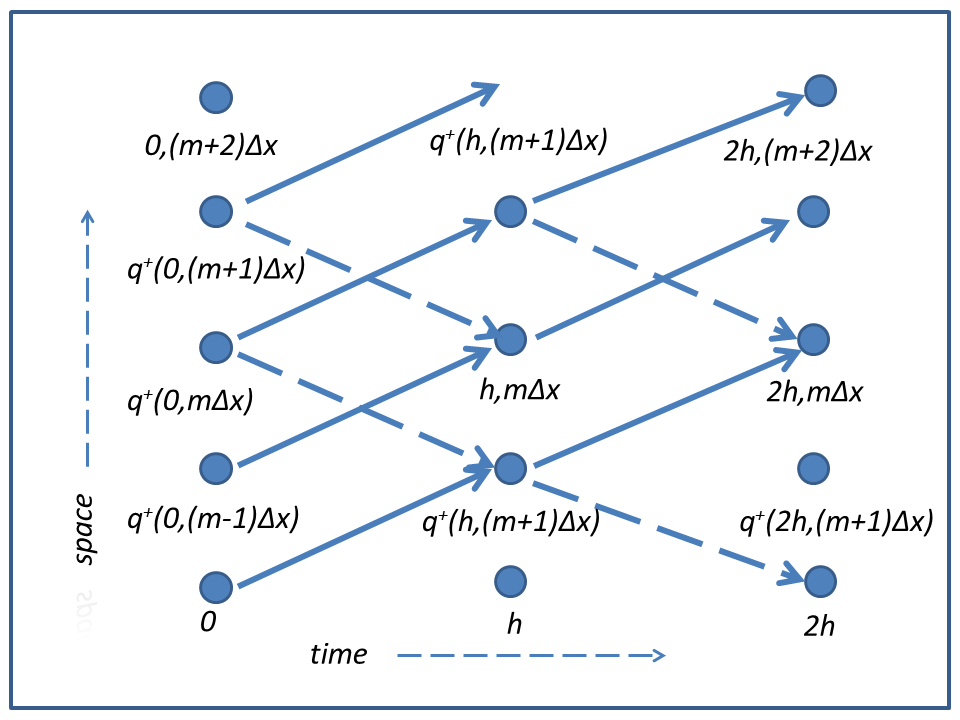}
\caption{Initial Grid / Tree Transitions}
\label{NodeTransitionOrigv2}
\end{figure}

To start the solution to \eqref{TRANSITION1} consider
Figure \ref{NodeTransitionOrigv2} where the positive
probability of transition is given solid and the
downward probabilities are gridded.  The first line
shows the starting zero line and the starting
probabilities $q^+(0,m\Delta x)$ for all $m$.
Simultaneously there is set of initial probabilities
in the downward direction $q^-(0,m\Delta x)$ for all
$m$.  The second line $q^+(h,m\Delta x)$ is
constructed from a combination of both $q^+(0,m\Delta
x)$ and $q^-(0,m\Delta x)$ using \eqref{TRANSITION1}.
From \eqref{TRANSITION43} it is clear that neither
$q^+(0,m\Delta x)$ or $q^-(0,m\Delta x)$ are proper
distributions but together they are to add to one.
Hence
\begin{align*}
&\rho(0,m\Delta x) = q^+(0,m\Delta x) + q^-(0,m\Delta x) \text{   all $m$}
\\
&q^+(0,m\Delta x) \geq 0 \text{   all $m$}
\\
&q^-(0,m\Delta x) \geq 0 \text{   all $m$}
\\
&\sum_{m}\left(q^+(0,m\Delta x) + q^-(0,m\Delta x)\right)=1
\end{align*}
is the appropriate starting condition for the process.

If a particle residing in $x$ at time $t$ steps to
$x+\Delta$ at time $t+h$ it adopts a velocity of
$v^+=\Delta x /h=c$. Similarly if the particle steps
down to $\Delta x$ it adopts a velocity of
$v^-=-\Delta x /h=-c$. Clearly then the + and the - in
the marginal densities $q^\pm(t,x)$ also refer to the
speed of the particle passing through $x$.  If there
are bigger steps from $x$ where it jumps many nodes
there could be speeds of $jc,j=...,-2,-1,0,1,2,...$ as
multiples of the one-step particle velocity.

\section{Numerical Examples}
\label{Numerics1}
\subsection*{First Example, small transition
probabilities} The first numerical solution is
presented in Figures \ref{SmallMDistribution1} and
\ref{SmallMDistribution1a} presenting the numerical
probability density solution of equation
\eqref{TRANSITION1} for the case where the transition
probabilities are very small ($\alpha=0.006,
\beta=0.006$ per time step).  The size of $\Delta
x=0.3$ and the timestep $h=0.003$ so the distribution
tree has a relatively high speed of $0.3 / 0.003 = \pm
100$. The initial probability distribution at time 0
stretches from -6.9 to 6.9 over some 46 nodes.

For this case we use the initial distributions
$q^+(0,x)$, $q^-(0,x)$ assuming a set of discrete
Gaussian distributions so that
\begin{align*}
&q^+(0,x)=q^-(0,x)= c_0e^{-\frac{x^2}{2\sigma^2}}
\\
&\rho(0,m\Delta x) =q^+(0,m\Delta x) +q^-(0,m\Delta x)
\\
&\sum_{m}\rho(0,m\Delta x) = 1
\end{align*}
where the constant $c_0$ has been chosen so that
$\rho(0,m\Delta x)$ is a proper discrete distribution
in the discrete parameter $m$.  This distribution of
$\rho(0,x)$ is shown in the center of Figure
\ref{SmallMDistribution1} for a standard deviation
equal to 0.6.

The calculation grid in this is a tree starting at
zero steadily enlarging for some 150 timesteps where
the final time will be $t=0.45=150*0.03$.  The effect
of the small transition probabilities suggests that
half the initial probability solutions is sent
symmetrically up the grid with speed $100$ and half
the initial distribution is sent down the grid at
$-100$.

With 46 nodes the symmetric original initial
distribution shows a distribution of $-6.9,6.9$ as the
central distribution in Figure
\ref{SmallMDistribution1} shows.  With the fact that
the initial tree is $-6.9,6.9$ this means that the
final set of nodes reaches $-45-6.9,45+6.9$ =
$-51.9,51.9$. which is exactly the boundary shown for
the left side and right side distribution in Figure
\ref{SmallMDistribution1}. Notice that the two left
and right distribution are very similar to the initial
distribution but not exactly the same while the amount
of probability in this distribution is less than half
the initial probability density.

\begin{figure}
\centering
\includegraphics[width=4in]{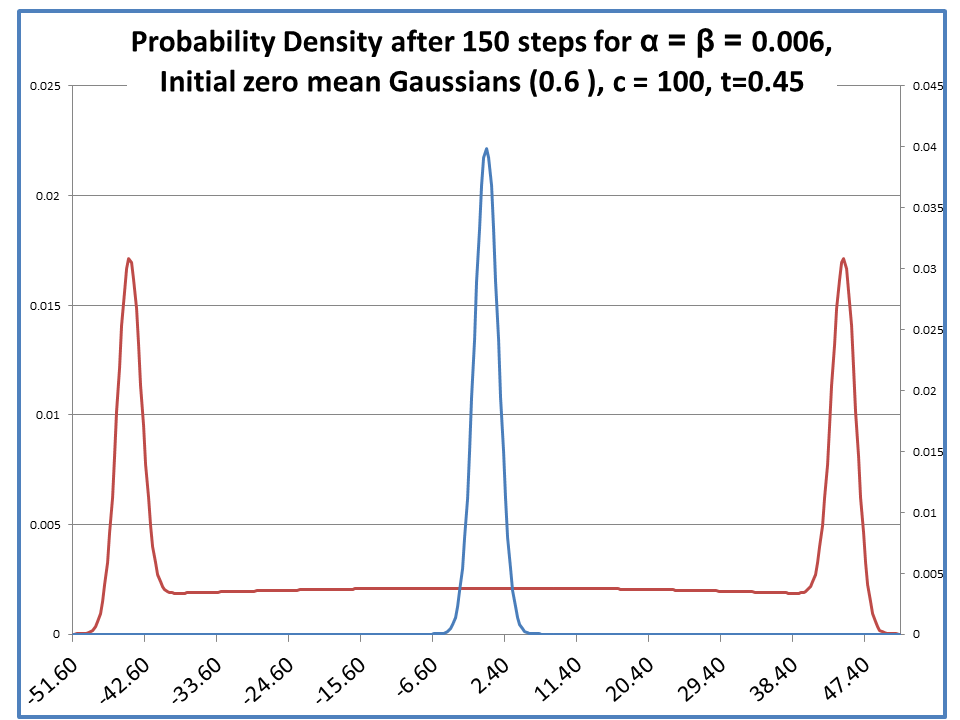}
\caption{Initial, Final Node Densities, Small Transitions}
\label{SmallMDistribution1}
\end{figure}

\begin{figure}
\centering
\includegraphics[width=4in]{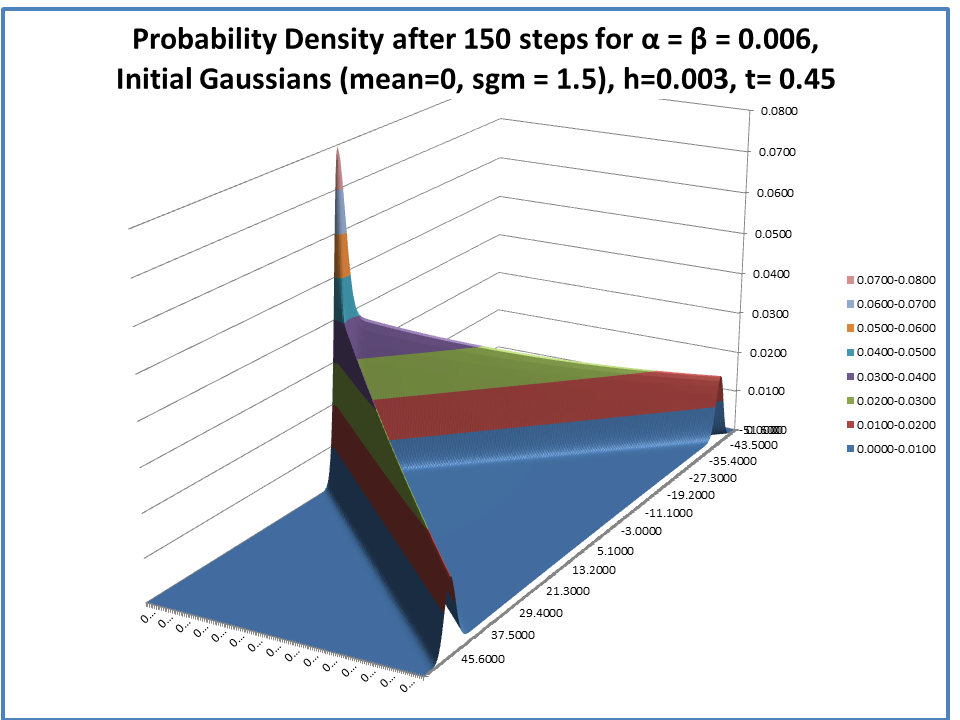}
\caption{3D Node Transition Travel, Small Transitions}
\label{SmallMDistribution1a}
\end{figure}

Once the initial distributions have been determined
the distribution is calculated from the difference
equation \eqref{TRANSITION1} and Figure
\ref{SmallMDistribution1} shows the final distribution
after 150 time steps and the initial distribution
$\rho(0,x)$. Figure \ref{SmallMDistribution1a}
presents a three dimensional picture of the
distribution changes.  The initial distribution
contains all the probability which is then split in
two, half of it traveling up and 100 and part of it
traveling down at -100.  The end distributions are
slightly wider than the initial distribution.  Notice
that there is also a certain amount of probability
assigned to the interval in between the extremes,
compare Figures \ref{SmallMDistribution1} and
\ref{SmallMDistribution1a}.

Physically this example shows the distribution of
particles that start in about 46 nodes between -6.9
and 6.9 on the real line and then step upward or
downward with about equal probability.  Once the
particles start moving up / down the grid the change
that they change direction is small.

\begin{figure}
\centering
\includegraphics[width=4in]{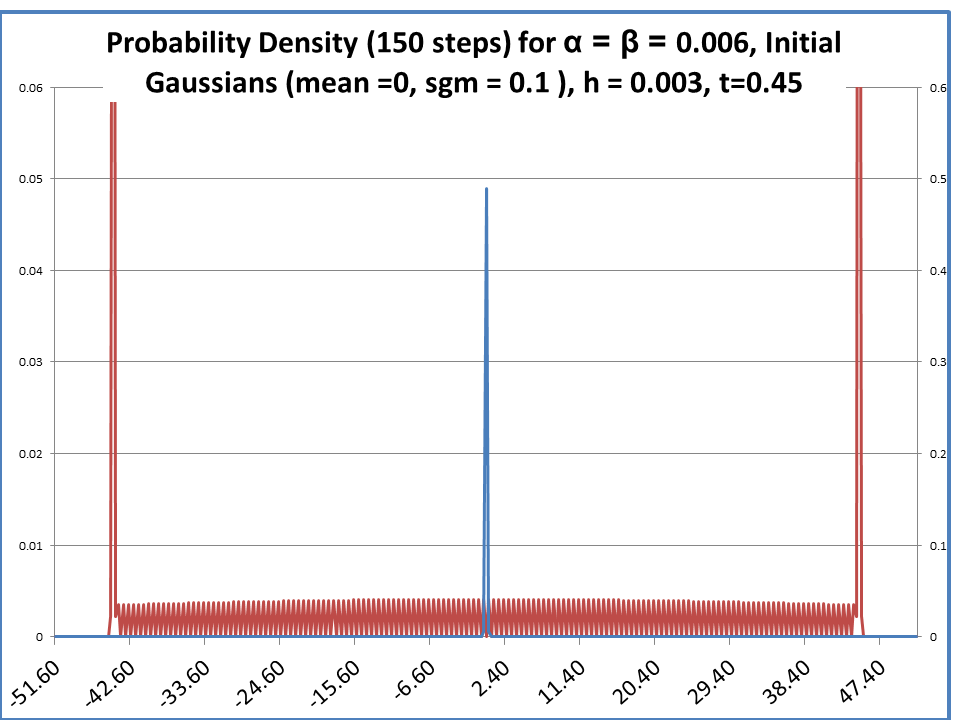}
\caption{Extreme Narrow Initial Probability Concentration}
\label{SmallMDistribution2}
\end{figure}

\subsection*{Second Example, small transition probabilities, pointed initial
conditions} A more extreme concentration version of
the previous example can be created by taking the
initial distribution equal to a point weight.  The
simulation is exactly the same as in the example above
but the initial distribution now has a standard
deviation of 0.1.  The initial distribution
distribution and the 150 step final result is shown in
Figure \ref{SmallMDistribution2}. Notice that the end
distribution (left and right side of Figure) have been
scaled up (right hand side of the Figure) showing the
final distributions of the initial conditions as in
the previous example.  However also notice that the
final distribution between the extremes is now a clear
zig-zag pattern of values.

This Figure erratic distribution pattern of
probability over the final distribution is due to the
fact that the particles can only step up or down.
Starting in node 0 a particle can only reside on node
+1 or -1 after one time step but not on node 0.
Similarly after 2 steps the particle can only reside
on node +2, 0, -2 but not on nodes +1 or -1.  The
final result is a superimposed interference pattern
which is an artifact of the fact that the particles
are not allowed to remain on certain nodes and the
concentrated initial distribution. As a result there
is a noticeable difference in the distribution between
adjacent nodes and the difference in value creates the
interference. Taking a wider initial distribution
removes this mixture and recreates a more continuous
final distribution.

Interestingly the interference patterns depend on the
initial distributions and the transition
probabilities.  Given the small transition
probabilities the interference decreases with
increasing initial distribution variance around the
0.25 - 0.3 standard deviation.  This was determined by
looking at Figure \ref{SmallMDistribution2} while
varying the standard deviation of the initial
distribution.

\begin{figure}
\centering
\includegraphics[width=4in]{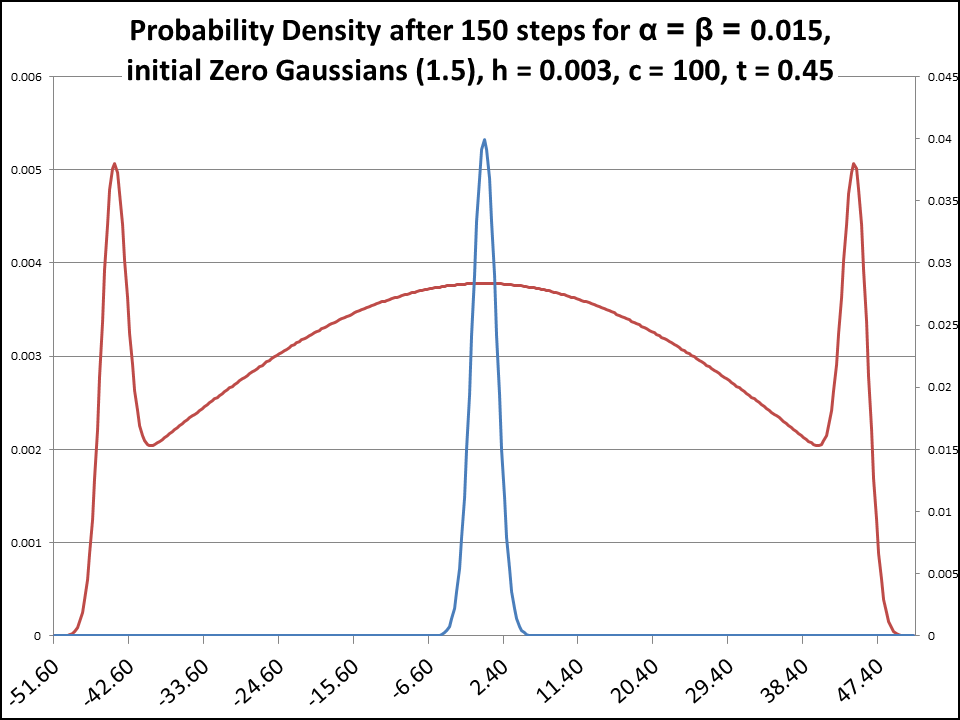}
\caption{Initial, Final Node Densities, Mixed Transitions}
\label{SmallMDistribution3}
\end{figure}

\begin{figure}
\centering
\includegraphics[width=4in]{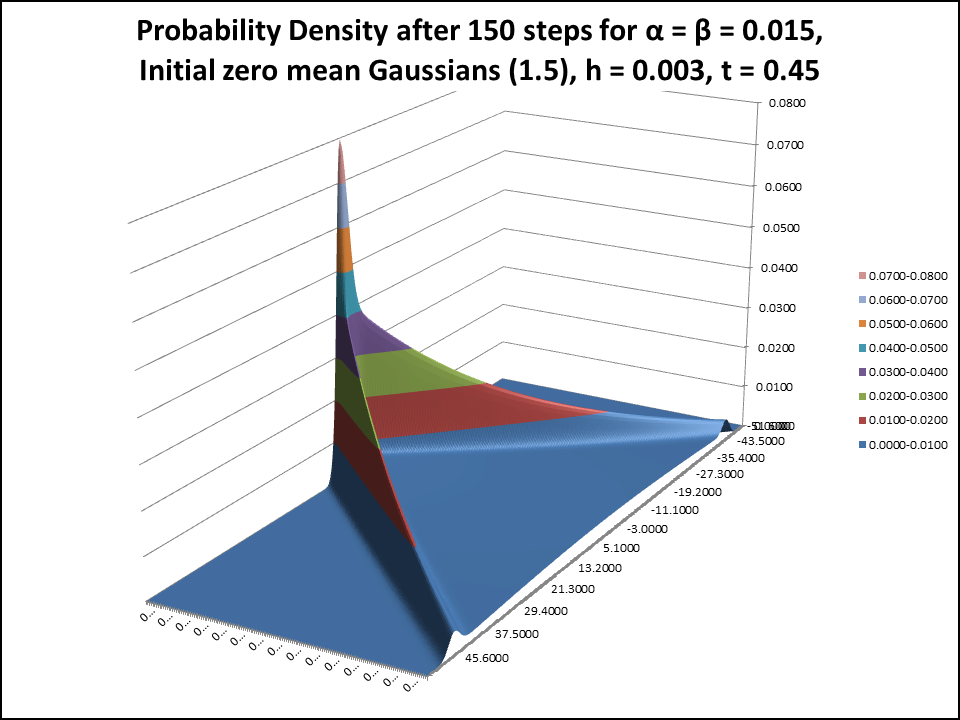}
\caption{3D Node Transition Travel, Mixed Transitions}
\label{SmallMDistribution3a}
\end{figure}

\subsection*{Third Example, mixed transition probabilities,
symmetric} In the next example the transition
probabilities are much larger so that part of the
probability travels to the edges following the speed
requirements but the remaining part of the
distribution is settled in the middle.  Figures
\ref{SmallMDistribution3} and
\ref{SmallMDistribution3a} show a more balanced case
again where the standard deviation of the initial
distributions is put back to 1.5 but now the
transition probabilities are increased to
$\alpha=0.015=1.5\%$, $\beta=0.015=1.5\%$.  Much of
the final distribution is away from the solution bound
-45.00 and 45.00 but some of the particles still end
there.

Comparing Figures \ref{SmallMDistribution3a} with
Figure \ref{SmallMDistribution1a} it is clear that a
larger part of the distribution is located between the
extreme nodes.  Also the size of the distributions at
the -51.9 and 51.9 extreme points is now quite a bit
smaller reflecting the probability diffusion into the
region between the extremes.  As Figure
\ref{SmallMDistribution3} suggests the final
distribution looks vaguely Gaussian but has
substantial "ears" at the extremes.

\subsection*{Fourth Example, large transition probabilities, non-symmetric}
Figures \ref{SmallMDistribution4},
\ref{SmallMDistribution4a} show an example of the grid
building for the case where the transition
probabilities are much bigger and also not
symmetrical. In this case the extreme distributions
will disappear as they can only be obtained if the
particles do not deviate from the straight path which
is unlikely given the 150 steps. Also the recombining
steps generate a distribution that is much more
centered in the middle while showing a curve in the
direction of the largest transition parameter.

Figures \ref{SmallMDistribution4} shows the initial
and final 150 step distribution.  The final
distribution looks Gaussian but has increased in
standard deviation and moved off-center. The initial
distribution was centered around zero and then
migrated to 10.99 at time 0.45 while the standard
deviation started with 1.5 as per our initial
distribution and then changed to about 16.2 at the
final time $t = 0.45$. Due to the distribution
spreading there is always an increase in the standard
deviation and in the next Section we will show how
this depends on time and probabilities.

The migration of the mean is difficult to estimate and
is not equal to $150*(\alpha-\beta)*\Delta x$.  The
$\alpha$ and $\beta$ determine the curvature of the
path but they say nothing about the mean motion or the
"up" or "down" probability.  This "up" and "down"
probability effectively determines the movement of the
mean but they cannot be easily determined without
calculation and for these cases they are time
dependent.

\begin{figure}
\centering
\includegraphics[width=4in]{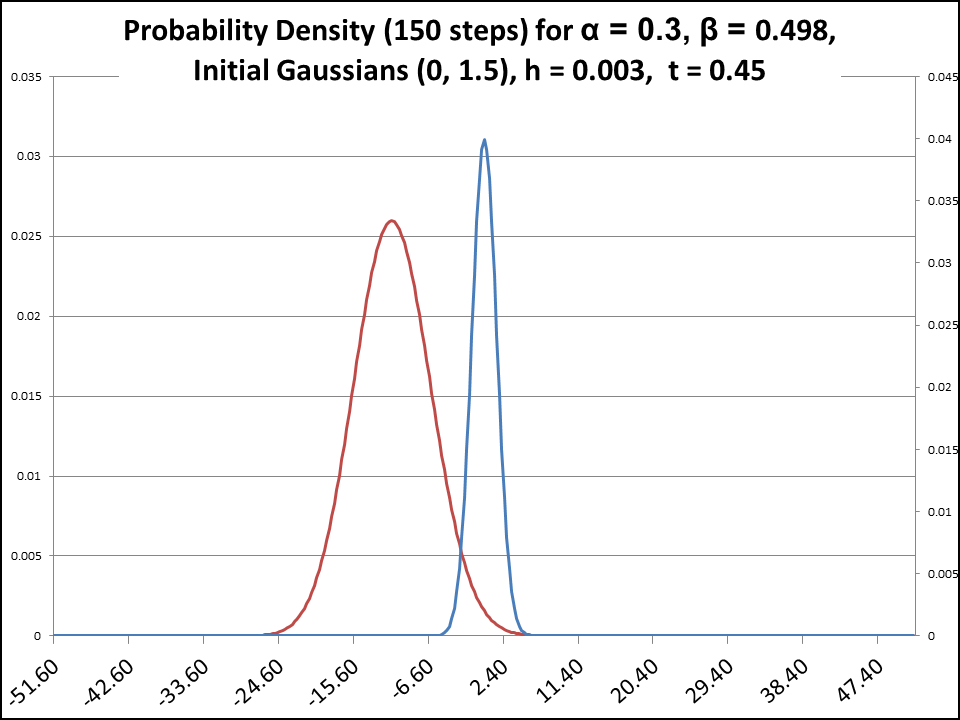}
\caption{Initial, Final Node Densities, Large Transitions}
\label{SmallMDistribution4}
\end{figure}

\begin{figure}
\centering
\includegraphics[width=4in]{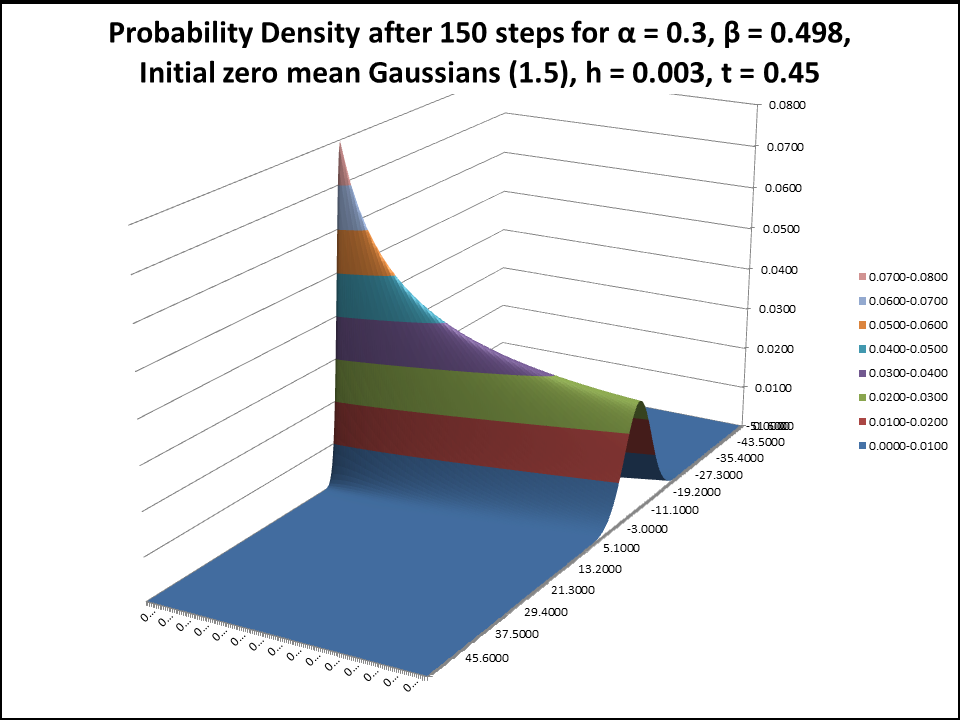}
\caption{3D Node Transition Travel, Large Transitions}
\label{SmallMDistribution4a}
\end{figure}

\section{Continuous Case}
\label{Continuous1}
\subsection*{Continuous Equation}
\label{ContinuousSection} Assuming an appropriate
limit leaving $\Delta / h \rightarrow c$ and in
proportion it is possible to transform equation
\eqref{TRANSITION1} to a continuous equation.  Assume
that $\Delta / h = c$ which means that the speed of
the particle is always constant.  So then a particle
moves only with speed $c$ or equal $-c$.  Additionally
assume that the probabilities $\alpha$ and $\beta$
behave as rates getting smaller with $h$ so that
\begin{align*}
\alpha(t,x)&\rightarrow \alpha(t,x)h
\\
\beta(t,x)&\rightarrow \beta(t,x)h.
\end{align*}
Then a two factor jump process can be constructed for
the position probability of the particle.

To apply these limits to equation \eqref{TRANSITION1}
insert $c$ and change the probabilities as rates.
Expand the $\Delta x$ terms to find that
\begin{align*}
\begin{pmatrix}
q^+(t+h,x) \\
q^-(t+h,x)
\end{pmatrix}=
\begin{pmatrix}
1-\alpha(t,x)h & \beta(t,x)h\\
\alpha(t,x)h & 1-\beta(t,x)h
\end{pmatrix}
\begin{pmatrix}
q^+(t,x)-\Delta x\frac{\partial}{\partial x}q^+(t,x)+O(\Delta^2 x) \\
q^-(t,x)+\Delta x\frac{\partial}{\partial x}q^-(t,x)+O(\Delta^2 x)
\end{pmatrix}
\end{align*}
so that
\begin{align*}
\begin{pmatrix}
q^+(t+h,x) \\
q^-(t+h,x)
\end{pmatrix}
=&
\begin{pmatrix}
1-\alpha(t,x)h & \beta(t,x)h\\
\alpha(t,x)h & 1-\beta(t,x)h
\end{pmatrix}
\begin{pmatrix}
q^+(t,x) \\
q^-(t,x)
\end{pmatrix}
\\
&+\Delta x
\begin{pmatrix}
1-\alpha(t,x)h & \beta(t,x)h\\
\alpha(t,x)h & 1-\beta(t,x)h
\end{pmatrix}
\begin{pmatrix}
-\frac{\partial}{\partial x}q^+(t,x) \\
\frac{\partial}{\partial x}q^-(t,x)
\end{pmatrix}.
\end{align*}

Notice that in the last part of the equation there are
terms $h\Delta x$ that can be ignored as they are
small.  Retaining the main terms yields
\begin{align*}
\begin{split}
\begin{pmatrix}
q^+(t+h,x)- q^+(t,x) \\
q^-(t+h,x)- q^-(t,x)
\end{pmatrix}
=& h\begin{pmatrix}
-\alpha(t,x) & \beta(t,x)\\
\alpha(t,x) &-\beta(t,x)
\end{pmatrix}
\begin{pmatrix}
q^+(t,x) \\
q^-(t,x)
\end{pmatrix}
\\
&+\Delta x
\begin{pmatrix}
-\frac{\partial}{\partial x}q^+(t,x) \\
\frac{\partial}{\partial x}q^-(t,x)
\end{pmatrix}
\end{split}
\end{align*}
so dividing by $h$ and calculating the limit yields
\begin{align}
\label{TRANSITION5}
\begin{split}
\frac{\partial}{\partial t}
\begin{pmatrix}
q^+(t,x) \\
q^-(t,x)
\end{pmatrix}
=& \begin{pmatrix}
-\alpha(t,x) & \beta(t,x)\\
\alpha(t,x) &-\beta(t,x)
\end{pmatrix}
\begin{pmatrix}
q^+(t,x) \\
q^-(t,x)
\end{pmatrix}
+c
\frac{\partial}{\partial x}
\begin{pmatrix}
-q^+(t,x) \\
q^-(t,x)
\end{pmatrix}
\end{split}
\end{align}
taking $\Delta \downarrow 0$,$h \downarrow 0$ and
setting $\Delta / h=c$. The term $h^2$ is ignored as
they are an order of magnitude smaller.

Reorganizing this yields
\begin{align}
\begin{split}
\label{TRANSITION4}
\frac{\partial q^+(t,x)}{\partial t} +
c\frac{\partial q^+(t,x)}{\partial x}+\alpha(t,x)q^+(t,x) &=
\beta(t,x)q^-(t,x)
\\
\frac{\partial q^-(t,x)}{\partial t} - c\frac{\partial
q^-(t,x)}{\partial x}+\beta(t,x)q^-(t,x) &=
\alpha(t,x)q^+(t,x)
\end{split}
\end{align}
which is a set of coupled convection equations flowing
the probability over the grid.  This type of equation
looks like a hyperbolic one-dimensional Telegraph
Equation for the behavior of voltage and current waves
in a lossy transmission line though the signs are
different ~\cite{HARRINGTON} or ~\cite{MITTAL1}.  In
this case there are only initial conditions in the
form of initial densities $q^+(0,x)$, $q^-(0,x)$ and
there are usually no Dirichlet type boundary
conditions (time dependent fixed $x$ boundaries).
Dirichlet boundaries only arise if the probability
flow is restricted in state over time.

\subsection*{Telegraph Equation}
Equation \eqref{TRANSITION4} can be recast in a two
dimensional Telegraph equation. As before let
$\rho(t,x)=q^+(t,x)+q^-(t,x)$ and define
$\phi(t,x)=q^+(t,x)-q^-(t,x)$ then equations
\eqref{TRANSITION5} and \eqref{TRANSITION4} can be
rewritten as
\begin{align*}
\begin{split}
&\frac{\partial}{\partial t}\rho(t,x) = -c\frac{\partial}{\partial
x}\phi(t,x)
\\
&\frac{\partial}{\partial t}\phi(t,x) =-c\frac{\partial}{\partial
x}\rho(t,x)-2\alpha(t,x)q^+(t,x)+2\beta(t,x)q^-(t,x)
\end{split}
\end{align*}
or substituting $q^+(t,x)=(\rho(t,x)+\phi(t,x))/2$ and
$q^-(t,x)=(\rho(t,x)-\phi(t,x))/2$ this reduces to
\begin{align}
\label{TRANSITION7a}
&\frac{\partial}{\partial t}\rho(t,x) = -c\frac{\partial}{\partial
x}\phi(t,x)
\\
\label{TRANSITION7b}
&\frac{\partial}{\partial t}\phi(t,x) =-c\frac{\partial}{\partial
x}\rho(t,x)-\epsilon(t,x)\rho(t,x)-\gamma(t,x)\phi(t,x)
\end{align}
where $\gamma(t,x)=\alpha(t,x)+\beta(t,x)$ and
$\epsilon(t,x)=\alpha(t,x)-\beta(t,x)$.

If $\alpha(t,x)=\alpha, \beta(t,x)=\beta$ (reducing
$\gamma(t,x)$ and  $\epsilon(t,x)$ to constants) this
equation can be further reduced. Denote
$\gamma(t,x)=\gamma, \epsilon(t,x)=\epsilon$ and
taking the time derivative  in \eqref{TRANSITION7a}
equation into the second equation \eqref{TRANSITION7b}
yields
\begin{align*}
\begin{split}
\frac{\partial^2}{\partial t^2}\rho(t,x) =&
-c\frac{\partial}{\partial x}\left(-c\frac{\partial} {\partial
x}\rho(t,x)-\epsilon\rho(t,x)-\gamma \phi(t,x)\right)
\\
=& c^2\frac{\partial^2}{\partial
x^2}\rho(t,x)+c\epsilon\frac{\partial} {\partial x}\rho(t,x)+c\gamma
\frac{\partial}{\partial x}\phi(t,x)
\end{split}
\end{align*}
which becomes
\begin{align}
\label{TRANSITION8}
\begin{split}
\frac{\partial^2}{\partial t^2}\rho(t,x)+ \gamma
\frac{\partial}{\partial t}\rho(t,x)= c^2\frac{\partial^2}{\partial
x^2}\rho(t,x)+c\epsilon\frac{\partial}{\partial x}\rho(t,x)
\end{split}
\end{align}
using \eqref{TRANSITION7a} on the last term in the
equation.

This is a two dimensional Telegraph equation which has
damping in time $t$ as well as in the spatial
coordinate $x$.  Notice various applications involving
voltage and current waves in ~\cite{HOSSEINI1},
~\cite{JAVIDI1} though these typically have Dirichlet
type additional conditions.  Also the Cauchy problem
for the Telegraph equation based on its simulation by
a one-dimensional Markov random evolution has a
similar form to ~\cite{SAMOILENKO1} using Cauchy
boundary conditions.  There are applications of
stochastic processes to biology that have generated
the Telegraph Equation with Cauchy boundaries
~\cite{CODLING1}.

The Telegraph equation in statistics have been
proposed in the literature to describe motions of
particles with finite velocities as opposed to
diffusion-type models. The first contribution in this
area goes back to ~\cite{GOLDSTEIN1}, ~\cite{KAC1}. In
~\cite{IACUS1} and in ~\cite{BOGUNA1} is shown the
Telegraph equations similar to \eqref{TRANSITION4} and
\eqref{TRANSITION8} for local probabilities.  This is
more an equation for the joint probabilities rather
than the conditional probabilities used in this
manuscript. For a Markov process representation using
Telegraph jump processes and market models see
~\cite{RATANOV1}.

This equation can be reduced to the damped
Klein-Gordon as we will see below. Notice that in this
case the initial condition requires an initial
function for $\rho(t,x)$ as well as $\phi(t,x)$
satisfying \eqref{TRANSITION7a} and
\eqref{TRANSITION7b}.

\subsection*{Klein Gordon Equation}
It is relatively straightforward to change equation
\eqref{TRANSITION8} into a Klein-Gordon equation as
follows.

\begin{thm}
\label{THEOREM1} To reduce equation
\eqref{TRANSITION8} further assume that the solution
can be written as
\begin{align}
\label{TRANSITION19}
\rho(t,x)=e^{-\frac{\epsilon}{2c}x-\frac{\gamma}{2}t}\psi(t,x)
\end{align}
for constant $\epsilon, \gamma$ then
\begin{align}
\begin{split}
\label{TRANSITION18}
\frac{\partial^2 \psi(t,x)}{\partial t^2}
&=c^2\frac{\partial^2 \psi(t,x)}{\partial
x^2}+\eta^2\psi
\end{split}
\end{align}
with $4\eta^2=4\alpha\beta=(\gamma^2-\epsilon^2)$.
\end{thm}
\begin{proof}
Writing the state probability as
\begin{align*}
\rho(t,x)=e^{Ax+Bt}\psi(t,x)
\end{align*}
then for the first two terms in \eqref{TRANSITION8}
\begin{align*}
&\gamma\frac{\partial \rho(t,x)}{\partial t}= \gamma
e^{Ax+Bt}\left(\frac{\partial \psi(t,x)}{\partial t} +B\psi\right)
\\
&\frac{\partial^2 \rho(t,x)}{\partial t^2}=
e^{Ax+Bt}\left(\frac{\partial^2 \psi(t,x)}{\partial
t^2}+2B\frac{\partial \psi(t,x)}{\partial t} +B^2\psi\right)
\end{align*}
or
\begin{align*}
\frac{\partial^2 \rho(t,x)}{\partial t^2}+ \gamma\frac{\partial
\rho(t,x)}{\partial t}&= e^{Ax+Bt}\left(\frac{\partial^2
\psi(t,x)}{\partial t^2}+(2B+\gamma)\frac{\partial
\psi(t,x)}{\partial t} +(B^2+\gamma B)\psi\right)
\\
&=e^{Ax+Bt}\left(\frac{\partial^2 \psi(t,x)} {\partial
t^2}-\frac{\gamma^2}{4}\psi\right)
\end{align*}
if we use $B=-\gamma/2$ to dispense the first
derivative.

Similarly for the two terms on the right
\begin{align*}
&c\epsilon\frac{\partial \rho(t,x)}{\partial x}= c\epsilon
e^{Ax+Bt}\left(\frac{\partial \psi(t,x)}{\partial x} +A\psi\right)
\\
&c^2\frac{\partial^2 \rho(t,x)}{\partial x^2}=
c^2e^{Ax+Bt}\left(\frac{\partial^2 \psi(t,x)}{\partial
x^2}+2A\frac{\partial \psi(t,x)}{\partial x} +A^2\psi\right)
\end{align*}
or
\begin{align*}
c^2\frac{\partial^2 \rho(t,x)}{\partial x^2}+
c\epsilon\frac{\partial \rho(t,x)}{\partial x}&=
e^{Ax+Bt}\left(c^2\frac{\partial^2 \psi(t,x)}{\partial
x^2}+(2c^2A+c\epsilon)\frac{\partial \psi(t,x)}{\partial x}
+Ac(Ac+\epsilon)\psi\right)
\\
&=e^{Ax+Bt}\left(c^2\frac{\partial^2 \psi(t,x)}{\partial
x^2}-\frac{\epsilon^2}{4}\psi\right)
\end{align*}
if again we use $A=-\epsilon/(2c)$ to remove the first
derivative.

These two results reduce equation \eqref{TRANSITION8}
to
\begin{align*}
\frac{\partial^2 \psi(t,x)}{\partial t^2}
&=c^2\frac{\partial^2 \psi(t,x)}{\partial
x^2}+\frac{\gamma^2-\epsilon^2}{4}\psi
=c^2\frac{\partial^2 \psi(t,x)}{\partial x^2}+\eta^2\psi
\end{align*}
with the definition for $\eta$ above.
\end{proof}

There are many solutions to this equation depending on
the type of applications.  One interesting example to
be further discussed below is that the general
solution to \eqref{TRANSITION18} equals
\begin{align*}
\psi(t,x) &= AI_0(\xi) + BK_0(\xi)
\\
\xi&=\frac{\eta}{c}\sqrt{K(t,x)}
\\
K(t,x)&=c^2t^2-x^2
\end{align*}
where $I_0(.),K_0(.)$ are first order modified Bessel
functions. One particularly interesting fact is that
if $\psi(t,x)$ is an equation to \eqref{TRANSITION18}
then
\begin{align}
\label{TRANSITION44}
\psi^*\left(\frac{t-xv/c^2}{\sqrt{\vphantom(1-v^2/c^2)}},
\frac{x-vt}{\sqrt{\vphantom(1-v^2/c^2)}}\right)
\end{align}
is a solution as well for an arbitrary velocity $v$.

Imagine a reasonably large mass such that $\gamma<<1$
moving with speed $v$ then using \eqref{TRANSITION44}
we have that
\begin{align*}
\rho(t,x)=e^{-\frac{\epsilon}{2c}x-\frac{\gamma}{2}t}\psi^*(t,x)
\approx \psi^*(t,x)
\end{align*}
is a probability density hugging the straight path of
the underlying mass moving at speed $v$.  Here
$\gamma$ is small because a larger mass has little
dispersion so that $\gamma \approx 0$ and the $x$ term
in the exponential disappears.  Also it is assumed
that $\epsilon$ is small assuming a very large speed
of $c$.  So $\epsilon / c \approx 0$ and the $x$ term
from the exponential disappears and the approximation
holds.

This Klein-Gordon wave equation is commonly used in
relativistic quantum mechanics to find the spin-less
free particle for a wave function. Equation
\eqref{TRANSITION18} applies to the actual equation
after adjustment and is related to tachyons
\cite{FEINBERG} or transport equations \cite{ANNO1}.

\subsection*{Speed, Forward and Backward}
To find out
about the average position of a particle going through
$x$ we look into the probability exiting from one
node. Using Bayes on equation \eqref{TRANSITION1b} and
\eqref{TRANSITION1} any particle leaving state $x$ at
time $t$ has a velocity of $\pm c$ with probability
\begin{align*}
P[x(t+h)=x+\Delta x | x(t)=x] = \frac{P[x(t+h)=x+\Delta x,x(t)=x]}{P[x(t)=x]}
=\frac{q^+(t,x)}{\rho(t,x)}
\\
P[x(t+h)=x-\Delta x | x(t)=x] = \frac{P[x(t+h)=x-\Delta x,x(t)=x]}{P[x(t)=x]}
=\frac{q^+(t,x)}{\rho(t,x)}
\end{align*}
because $q^\pm(t,x)=P[x(t+h)=x\pm \Delta x,x(t)=x]$ is
the joint probability, see \eqref{TRANSITION1},
\eqref{TRANSITION1b}.

Hence the expectation of the exiting speed of the
particle conditional on residing in $x$ at time $t$
equals
\begin{align}
\label{VELOCITY}
v(t,x)=c\frac{(q^+(t,x)-
q^-(t,x))}{\rho(t,x)}=c\frac{\phi(t,x)}{\rho(t,x)}
\end{align}
where $\phi(t,x)=q^+(t,x)- q^-(t,x)$.  Hence the
difference between the two densities
$q^+(t,x)$,$q^-(t,x)$ also has a physical
interpretation.

For the backward velocity consider the steps in the
grid ending in $x(t)=x$ before taking the average. So
by definition using Bayes again
\begin{align*}
P[x(t-h)=x+\Delta x | x(t)=x] = \frac{P[x(t-h)=x+\Delta x,x(t)=x]}{P[x(t)=x]}
=\frac{q^+(t-h,x+\Delta x)}{\rho(t,x)}
\\
P[x(t-h)=x-\Delta x | x(t)=x] = \frac{P[x(t-h)=x-\Delta x,x(t)=x]}{P[x(t)=x]}
=\frac{q^+(t-h,x-\Delta x)}{\rho(t,x)}
\end{align*}
now concentrating on the probability of ending up in
$x$ coming from $x+\Delta x$ with velocity $-c$ or
coming from $x-\Delta x$ with velocity $c$.

It is possible to define a backward velocity and an
acceleration using an inversion of equation
\eqref{TRANSITION1}.

\begin{thm}
\label{VELOCITY2} The hyperbolic two dimensional
hyperbolic equation \eqref{TRANSITION4}  which has
forward velocity \eqref{VELOCITY} also has an
acceleration equal to
\begin{align*}
a_t(x)&=\lim_{h\downarrow 0}\frac{v(t,x)-v^-(t,x)}{h}
=-c\epsilon(t,x)-v(t,x)\gamma(t,x)
\end{align*}
with $v^-(t,x)$ defined as the speed of the particle
that ends in $x$ - the backward velocity.
\end{thm}
\begin{proof}
Using the Bayes argument as before the backward
velocity becomes
\begin{align*}
v^-(t,x) = c\left(P[x(t-h)=x-\Delta x | x(t)=x]-P[x(t-h)=x+\Delta x | x(t)=x]\right)
\end{align*}
which equals
\begin{align*}
v^-(t,x) &= c\left(\frac{P[x(t-h)=x-\Delta x, x_t=x]}
{P[x_t=x]}-\frac{P[x(t-h)=x+\Delta x, x_t=x]}{P[x_t=x]}\right)
\\
&=c\left(\frac{q^+(t-h,x-\Delta x)}{\rho(t,x)}-\frac{q^-(t-h,x+\Delta x)}{\rho(t,x)}\right)
\\
&=\frac{c}{\rho(t,x)}\left(q^+(t-h,x-\Delta x)-q^-(t-h,x+\Delta x)\right).
\end{align*}

This can be simplified by means of inverting equation \eqref{TRANSITION1} since then
\begin{align}
\label{TRANSITION37}
\begin{split}
\begin{pmatrix}
q^+(t-h,x-\Delta x) \\
q^-(t-h,x+\Delta x)
\end{pmatrix}
&=
\begin{pmatrix}
1-\alpha(t-h,x)h & \beta(t-h,x)h\\
\alpha(t-h,x)h & 1-\beta(t-h,x)h
\end{pmatrix}^{-1}
\begin{pmatrix}
q^+(t,x) \\
q^-(t,x)
\end{pmatrix}
\\
&=\frac{1}{D_{t-h}}
\begin{pmatrix}
 1-\beta(t-h,x)h & -\beta(t-h,x)h\\
-\alpha(t-h,x)h & 1-\alpha(t-h,x)h
\end{pmatrix}
\begin{pmatrix}
q^+(t,x) \\
q^-(t,x)
\end{pmatrix}
\\
&=\frac{1}{D_{t-h}}
\begin{pmatrix}
 (1-h\beta(t-h,x))q^+(t,x) -h\beta(t-h,x)q^-(t,x) \\
-hq^+(t,x)\alpha(t-h,x) + (1-h\alpha(t-h,x))q^-(t,x)
\end{pmatrix}
\\
&=\frac{1}{D_{t-h}}
\begin{pmatrix}
 q^+(t,x)-h\beta(t-h,x)\rho(t,x) \\
q^-(t,x)-h\alpha(t-h,x)\rho(t,x)
\end{pmatrix}
\end{split}
\end{align}
with the matrix determinant
\begin{align*}
D_{t-h}&=
(1-h\beta(t-h,x))(1-h\alpha(t-h,x)) -h^2\beta(t-h,x)\alpha(t-h,x)
\\
&=\left(1-h(\beta(t-h,x)+\alpha(t-h,x))\right)=1-h\gamma(t-h,x).
\end{align*}

Substituting this into equation \eqref{TRANSITION37}
yields
\begin{align*}
v^-(t,x)&=\frac{c}{\rho(t,x)}\left(q^+(t-h,x-\Delta x)-q^-(t-h,x+\Delta x)\right)
\\
&=\frac{c}{D_{t-h}\rho(t,x)}\left( q^+(t,x)-q^-(t,x)+h(\alpha(t-h,x)-\beta(t-h,x))\rho(t,x)\right)
\\
&=\frac{v(t,x)+ch\epsilon(t-h,x)}{D_{t-h}}=\frac{v(t,x)+hc\epsilon(t-h,x)}{1-h\gamma(t-h,x)}
\end{align*}
using $\epsilon(t-h,x)=\alpha(t-h,x)-\beta(t-h,x)$.

Now expanding in rising orders of $h$
\begin{align}
\begin{split}
\label{TRANSITION26}
v^-(t,x)&=\frac{v(t,x)+ch\epsilon(t-h,x)h}{1-h\gamma(t-h,x)}
\\
&=(v(t,x)+ch\epsilon(t-h,x))(1+h\gamma(t-h,x) + ...)
\\
&=v(t,x)+c\epsilon(t-h,x)h+v(t,x)\gamma(t-h,x)h + ...
\end{split}
\end{align}
or in the limit
\begin{align*}
a_t(x)&=\lim_{h\downarrow 0}\frac{v(t,x)-v^-(t,x)}{h}
=-c\epsilon(t,x)-v(t,x)\gamma(t,x)
\end{align*}
using $\epsilon(t,x)=\lim_{h\downarrow
0}(\alpha(t-h,x)-\beta(t-h,x)$,
$\gamma(t,x)=\lim_{h\downarrow
0}(\alpha(t-h,x)+\beta(t-h,x))$.
\end{proof}

This is therefore a practical definition of
acceleration through point $x$.  Notice that the
matters of limits and continuity for $\alpha(t,x)$,
$\beta(t,x)$ are technical issues which have been
ignored in this proof.

\section{Solutions}
\label{Solutions1}
Equation \eqref{TRANSITION4}, the
Telegraph equation \eqref{TRANSITION8} and the
Klein-Gordon equations \eqref{TRANSITION18} have
solutions for the case where there are specific
initial Cauchy conditions.  This involves
distributions for the initial condition and the
initial velocity.

The situation here is slightly different from the
usual applications in electromagnetism, investment or
previous probabilistic studies in that the
distribution is defined as the sum of $q^+(t,x)$,
$q^-(t,x)$ and no initial distribution is normally
provided for the Telegraph equation.

However, it is quite straightforwar5d to show the
following.

\begin{thm}
Using the notation from Theorem \ref{THEOREM1} with
$\eta^2 = (\beta^2 - \alpha^2)/4)$ with initial
conditions
\begin{align*}
\rho(0,x)=q^+(0,x)+q^-(0,x)
\end{align*}
then the solution to \eqref{TRANSITION4} and
\eqref{TRANSITION8} is given by
\begin{align}
\begin{split}
\label{TRANSITION45}
q^+(t,x)=e^{-\frac{\epsilon}{2c}x-\frac{\gamma}{2}t}\psi^+(t,x)
\\
q^-(t,x)=e^{-\frac{\epsilon}{2c}x-\frac{\gamma}{2}t}\psi^-(t,x)
\\
\rho(t,x)=q^+(t,x)+q^-(t,x)
\end{split}
\end{align}
where
\begin{align*}
\psi^\pm(t,x)=&\frac{\left(\psi^\pm(x+ct)+\psi^\pm(x-ct)\right)}{2}+
\frac{ct\eta}{2}\int_{x-ct}^{x+ct}\psi^\pm(z)
\frac{I_0(\frac{\eta}{c}\xi(t,x-z))}{\xi(t,x-z)}dz
\end{align*}
with
\begin{align*}
\psi^\pm(x)&=e^{\frac{\epsilon}{2c}x}q^\pm(0,x)
\\
\xi&=\sqrt{c^2t^2-x^2}.
\end{align*}
\end{thm}
\begin{proof}
Equation \eqref{TRANSITION18} presents the
Klein-Gordon equation for a function of the overall
space probability density $\rho(t,x)$.  This was
derived from \eqref{TRANSITION19} and
\eqref{TRANSITION8}. However, it is possible to use
the equations to show that the same Telegraph equation
\eqref{TRANSITION8} applies for $\phi(t,x)$. However
in that case the initial conditions would be
different.

Since both the addition $\rho(t,x)$ and the difference
$\phi(t,x)$ satisfy the Telegraph equation their
difference and their addition also satisfies the
Telegraph equation \eqref{TRANSITION8}.  Since
$q^+(t,x)=(\rho(t,x)+\phi(t,x))/2$,
$q^-(t,x)=(\rho(t,x)-\phi(t,x))/2$ also satisfy the
Telegraph equation once transformed with the
exponential \eqref{TRANSITION19} relation they also
satisfy the Klein-Gordon equation.  Hence $\rho(t,x),
\phi(t,x), q^+(t,x),q^-(t,x)$ all satisfy the
Klein-Gordon equation though with different
conditions.

The solution to the Klein-Gordon equation under the
Cauchy conditions is
\begin{align*}
\psi(t,x)=&\frac{(\psi(x+ct)+\psi(x-ct))}{2}+
\frac{ct\eta}{2}\int_{x-ct}^{x+ct}\psi(\xi)\frac{I_0(\frac{\eta}{c}\sqrt{\vphantom(c^2t^2-(x-\xi)^2})}
{\sqrt{\vphantom(c^2t^2-(x-\xi)^2}}d\xi
\\
&+\frac{1}{2c}\int_{x-ct}^{x+ct}
\theta(\xi)K_0(m\sqrt{\vphantom(c^2t^2-(x-\xi)^2}))d\xi
\end{align*}
where
\begin{align*}
\psi(0,x)=\psi(x)
\\
\frac{\partial}{\partial t}\psi(0,x)=\theta(x).
\end{align*}

Given the fact that we do not have a conditional
derivative we can drop the $\theta(t,x)$ term above
and put in the appropriate initial condition for the
individual densities.  Using \eqref{TRANSITION19}
these initial condition are
\begin{align*}
q^+(0,x)&=e^{-\frac{\epsilon}{2c}x}\psi^+(x)
\\
&\text{or}
\\
q^-(0,x)&=e^{-\frac{\epsilon}{2c}x}\psi^-(x)
\end{align*}
hence inverting this yields
$\psi^\pm(0,x)=e^{\frac{\epsilon}{2c}x}q^\pm(0,x)$
which are the initial conditions provided in
\eqref{TRANSITION45}.

The resulting equation then is a solution to the
Klein-Gordon equation using the initial condition for
$q^+(0,x),q^-(0,x)$.  To find the solution for the
Telegraph equation we need the exponential as in
\eqref{TRANSITION19} for either initial case which
explains the equations in \eqref{TRANSITION45}. The
equations for $q^=(t,x)$, $q^-(t,x)$ are one factor
equations and so the procedure above should find a
unique solution.
\end{proof}

On the other hand if $\gamma$ and $c$ become large so
that $\gamma,c >> 1$ then equation \eqref{TRANSITION8}
can be rewritten as
\begin{align}
\label{TRANSITION9}
\begin{split}
\frac{\partial}{\partial t}\rho(t,x)=
\frac{c^2}{\gamma}\frac{\partial^2}{\partial x^2}\rho(t,x)+
\frac{c\epsilon}{\gamma}\frac{\partial}{\partial x}\rho(t,x) -
\frac{1}{\gamma}\frac{\partial^2}{\partial t^2}\rho(t,x)
\end{split}
\end{align}
and the last term in the limit vanishes (since $1/\gamma << 1$) to
yield
\begin{align*}
\begin{split}
\frac{\partial}{\partial t}\rho(t,x)=
\frac{c^2}{\gamma}\frac{\partial^2}{\partial x^2}\rho(t,x)+
\frac{c\epsilon}{\gamma}\frac{\partial}{\partial x}\rho(t,x)
\end{split}
\end{align*}
which is a diffusion equation with variance $\sigma^2
= \dfrac{c^2}{\gamma}$ and drift
$-\dfrac{c\epsilon}{\gamma}$.

In fact, if in addition $\alpha=\beta$ then
$\epsilon=\alpha-\beta=0$ and $\gamma=2\alpha$ so equation
\eqref{TRANSITION9} reduces to
\begin{align*}
\begin{split}
\frac{\partial}{\partial t}\rho(t,x)=
\frac{c^2}{\gamma}\frac{\partial^2}{\partial x^2}\rho(t,x)=
\frac{c^2}{2\alpha}\frac{\partial^2}{\partial x^2}\rho(t,x)
\end{split}
\end{align*}
which is the quintessential diffusion.

The result is not surprising since a large $\gamma$ indicates a very
high probability of switching velocities in any state $x$.  The
typical Brownian path has a very high (infinite) velocity (here
equal to $c$) but changes directions at a very high (infinitely
large) rate.

Some properties can be seen more or less from equation
\eqref{TRANSITION8}.
\begin{thm}
For small $\gamma$ the velocity of the mean remains
equal to its initial value while for large $\gamma$
and large $c$ the mean speed of the mean becomes
\begin{align*}
v_0=-\frac{c\epsilon}{\gamma}.
\end{align*}
Also
\begin{align*}
E[x^2(t)] \approx
\frac{2c^2t}{\gamma}+E[x(t)]^2
\end{align*}
for $\gamma >> 1$ and $\gamma t >>1$.
\end{thm}
\begin{proof}
Integrating equation \eqref{TRANSITION8} shows that
\begin{align*}
\frac{d^2}{ dt^2}\int_{-\infty}^{\infty} x\rho(t,x)dx+ \gamma
\frac{d}{dt}\int_{-\infty}^{\infty}x\rho(t,x)dx= -c\epsilon
\end{align*}
so that
\begin{align*}
\frac{d^2}{ dt^2}E[x(t)]+\gamma \frac{d}{dt}E[x(t)]=-c\epsilon.
\end{align*}

The solution to this is straightforward
\begin{align}
\label{TRANSITION10}
v_t = \frac{d}{dt}E[x(t)]=-\frac{c\epsilon}{\gamma}+
\left(v_0+\frac{c\epsilon}{\gamma}\right)e^{-\gamma
t}
\end{align}
where
\begin{align*}
v_0 = \frac{d}{dt}E[x(0)]
\end{align*}
a simple reference to the initial velocity.  Clearly
then the particle has initial velocity $v_0$ but after
a fair time $\gamma t <<1$ we get
$v_\infty=-\frac{c\epsilon}{\gamma}$.

Integrating \eqref{TRANSITION10} once once more yields
\begin{align*}
E[x(t)]&=E[x(0)]-\frac{c\epsilon t}{\gamma}+
\left(v_0+\frac{c\epsilon}{\gamma}\right)\frac{1-e^{-\gamma
t}}{\gamma}
\\
&\approx E[x(0)]-\frac{c\epsilon t}{\gamma} \approx E[x(0)] + v_\infty t
\end{align*}
as long as $\gamma >> 1$ and $\gamma t >> 1$.

Equation \eqref{TRANSITION8} also shows that
\begin{align*}
\frac{d^2}{ dt^2}\int_{-\infty}^{\infty} x^2\rho(t,x)dx+ \gamma
\frac{d}{dt}\int_{-\infty}^{\infty}x^2\rho(t,x)dx= 2c^2-2c\epsilon
E[x(t)]
\end{align*}
so that
\begin{align}
\label{TRANSITION17}
\frac{d}{dt}E[x^2(t)]=e^{-\gamma t}
E[x^2(0)]+\frac{2c^2}{\gamma}\left(1-e^{-\gamma
t}\right)-2c\epsilon\int_0^te^{-\gamma (t-s)}E[x(s)]ds
\end{align}

The last term in this equation becomes
\begin{align*}
\int_0^te^{-\gamma (t-s)}E[x(s)]ds
&\approx \int_0^te^{-\gamma
(t-s)}\left(E[x(0)]-\frac{c\epsilon s}{\gamma}\right)ds
\\
&=\frac{E[x(0)]}{\gamma}\left(1-e^{-\gamma
t}\right)-\frac{c\epsilon}{\gamma}\left(\frac{t}{\gamma}-\frac{1-e^{-\gamma
t}}{\gamma^2}\right)
\\
&\approx \frac{E[x(0)]}{\gamma} -\frac{c\epsilon t}{\gamma^2}
\end{align*}
so then equation \eqref{TRANSITION17} becomes
\begin{align*}
\frac{d}{dt}E[x^2(t)]
&\approx \frac{2c^2}{\gamma}-\frac{2c\epsilon
E[x(0)]}{\gamma} +\frac{2c^2\epsilon^2 t}{\gamma^2}
\\
&\approx \frac{2c^2}{\gamma} - \frac{2c\epsilon}{\gamma}\left(
E[x(0)] - \frac{c\epsilon t}{\gamma}\right)
\\
&\approx \frac{2c^2}{\gamma} + v_\infty\left(
E[x(0)] + v_\infty t\right)
\\
&\approx \frac{2c^2}{\gamma} + \frac{d}{dt}E[x(t)]^2
\end{align*}
so that
\begin{align*}
E[x^2(t)] \approx
\frac{2c^2t}{\gamma}+E[x(t)]^2
\end{align*}
\end{proof}

This result shows that the distribution is consistent
with a Gaussian diffusion as the variance increases
with time while the remaining terms are of order
$1/\gamma$.

However, if $\alpha(t,x)$ and $\epsilon(t,x)$ depend
on the state $x$ and time $t$ explicitly then a
different approach is required.

\section{Multi-dimensional Case}
\label{MultiDimensional1}
\subsection*{Infinite Dimensional Equation}
It is easy to generalize the previous algorithm in
\eqref{TRANSITION1} by introducing both more states to
step into and more states to come from.  In the
multi-dimensional case the state $x(t)$ is connected
to many other states $x(t+h)=x+j\Delta x$. Define
\begin{align*}
q^j(t,x)=P[x(t+h)=x+j\Delta x,x(t)=x],\text{ for }j=...,-1,0,1,...
\end{align*}
as the set of joint speed and position probabilities
generalizing \eqref{INITIALTRANS1}.  To generalize
\eqref{TRANSITION1b} let probability being in
$x(t+h)=x$ and stepping to $x+k\Delta x$ at $t+h$
after traveling from $x(t)=x+j\Delta x$ at $t-h$ equal
\begin{align*}
\begin{matrix}
\omega_{jk}(t,x)=P_t[x+j\Delta x,x,x+k\Delta x]
\end{matrix}
\end{align*}
with $\sum_j\omega_{jk}=1$.  Then the equivalent to
\eqref{TRANSITION1} reduces to
\begin{align}
\label{TRANSITION21}
q^j(t+h,x) =\sum_k \omega_{jk}q^k(t,x-k\Delta x)
\end{align}
for $j,k=...,-1,0,1,...$.  This implies that
$q^j(t,x)$ is the joint distribution of being in
position $x_t=x$ and making a step the size $j\Delta
x,j=...,-1,0,1,...$.  Notice that this means that the
particle has a velocity of $v_j=j\Delta
x/h=jc,j=...,-1,0,1,...$.

As a result $\sum_j q^j(t,x)=\rho(t,x)$ is the
marginal probability of the particle residing in $x$
and so adding the equations in \eqref{TRANSITION21}
yields
\begin{align*}
\rho(t+h,x) = \sum_j q^j(t+h,x)
\end{align*}
so that
\begin{align*}
\sum_n\rho(t+h,n\Delta x) &= \sum_n \sum_j q ^j(t,x-n\Delta x)
\\
&=\sum_j \sum_n q ^j(t,x-n\Delta x)
\\
&=\sum_j q ^j(t,x)=1
\end{align*}
which shows that state probability is conserved.

To create an equation we now assume that the $\omega$
matrix becomes a rate similar to the $\alpha$, $\beta$
in Section \ref{ContinuousSection}.  Substitute
$\omega_{jk} \rightarrow \delta_{jk}+h\omega_{jk}$
then equation \eqref{TRANSITION21} becomes
\begin{align}
\label{TRANSITION46}
q^j(t+h,x) =\sum_k \left(\delta_{jk}+h\omega_{jk}\right)q^k(t,x-k\Delta x)
\end{align}
for all appropriate indices $j=...,-1,0,1,...$.  Now
$\omega_{jk}$ is a rate matrix which has positive
values for all off-diagonal elements and has negative
values on the diagonal with $\sum_j \omega_{jk}=0$.

To apply this to equation \eqref{TRANSITION21} expand
on the small terms $\Delta x$ and write
\begin{align*}
q^j(t+h,x)=\sum_k (\delta_{jk}+\omega_{jk}h) \left(q^k(t,x)
-k\Delta x\frac{\partial}{\partial x}q^k(t,x)\right)
\end{align*}
where $\sum_j \omega_{jk}=0$ for all $k$.

Retaining the main terms then yields
\begin{align*}
q^j(t+h,x)- q^j(t,x) =& h\sum_k\omega_{jk}q^k(t,x)
-j\Delta x\frac{\partial}{\partial x}q^j(t,x)
\end{align*}
and so finally in the limit
\begin{align}
\begin{split}
\label{TRANSITION30}
\frac{\partial}{\partial t} q^j(t,x)+v_j\frac{\partial}{\partial x} q^j(t,x)
=\sum_k \omega_{jk}q^k(t,x), \text{        $j=..., -1,0,1,...$}
\end{split}
\end{align}
were $v_j=j\Delta x/h$ as above and where the terms
$v_jh$ and $h^2$ can be ignored as they are an order
of magnitude smaller.  This is a set of coupled
advection equations.

\subsection*{Speed, Forward}
The definition of average speed $v(t,x)$ translates
directly from equation \eqref{VELOCITY}
\begin{align}
\label{VELOCITY1}
v(t,x)=\frac{\sum_j v_jq^j(t,x)}{\rho(t,x)}=c\frac{\sum_j jq^j(t,x)}{\rho(t,x)}
\end{align}
which also implies that
\begin{align*}
\sum_j v_jq^j(t,x)=v(t,x)\rho(t,x).
\end{align*}
Notice here we use the definition $\Delta x / h=c$.

Using this equation it is clear that
\begin{align*}
\frac{\partial}{\partial t}
\rho(t,x)+ \frac{\partial}{\partial x}\sum_j v_jq^j_{t}(x)=0
\end{align*}
so that
\begin{align}
\label{TRANSITION24}
\frac{\partial}{\partial t}
\rho(t,x)+ \frac{\partial}{\partial x}\left(v(t,x)\rho(t,x)\right)=0.
\end{align}

This equation is the continuity equation which is true
for any distribution no matter what the choices are
for velocities $v_k$, model size or otherwise.  The
remaining models depend on the choices of $\omega$ and
for some matrix configuration we can simulate the
Newtonian system.

\begin{thm}
\label{POTENTIAL1}
 Let the probability matrix in
\eqref{TRANSITION30} equal
\begin{align*}
w =
\begin{pmatrix}
\ddots & \alpha & 0 & 0
\\
\beta & -\lambda & \alpha & 0
\\
0 & \beta & -\lambda & \alpha
\\
0 & 0 & \beta & -\lambda
\\
0 & 0 & 0 & \beta
\end{pmatrix}
\end{align*}
with $\lambda=\alpha+\beta$ and let
\begin{align}
\label{TRANSITION31}
\alpha -\beta = \frac{1}{c}\frac{\partial V}{\partial x}
\end{align}
then
\begin{align*}
\frac{\partial^2}{\partial t^2}E\left[x(t)\right]
=E\left[\frac{\partial V}{\partial x}\right]
\end{align*}
so that the average motion of the particle follows
Newton's equation.
\end{thm}
\begin{proof}
Consider that equation \eqref{TRANSITION30} the per
velocity distribution and multiply each row with
$v_k$. Then sum over the equation to get
\begin{align*}
\frac{\partial}{\partial t} \sum_j v_j q^j_{t}(x)+
\frac{\partial}{\partial x}\sum_j v^2_jq^j_{t}(x)=
\sum_{jk} v_j\omega_{jk}q^k(t,x)
\end{align*}
with equation
\begin{align*}
w =
\begin{pmatrix}
\ddots & \alpha & 0 & 0
\\
\beta & -(\alpha+\beta) & \alpha & 0
\\
0 & \beta & -(\alpha+\beta) & \alpha
\\
0 & 0 & \beta & -(\alpha+\beta)
\\
0 & 0 & 0 & \beta
\end{pmatrix}.
\end{align*}

Take $v_k=k\Delta x / h=kc$ then for all $k$ we have
\begin{align*}
\sum_{j} v_j\omega_{jk} = \sum_{j} cj\omega_{jk}
&=c \left(\alpha (k-1) -(\alpha + \beta) k + \beta (k+1)\right)
\\
&=-c(\alpha -\beta)=-\frac{\partial V}{\partial x}
\end{align*}

As a result the right hand side of equation
\eqref{TRANSITION24} changes
\begin{align*}
\frac{\partial}{\partial t}
\sum_j v_j q^j_{t}(x)+ \frac{\partial}{\partial x}\sum_j v^2_jq^j_{t}(x)=
-c\epsilon\rho(t,x)=-\frac{\partial V}{\partial x}\rho(t,x)
\end{align*}

or using the continuity equation \eqref{TRANSITION24}
we find
\begin{align*}
\frac{\partial}{\partial t}
(v(t,x)\rho(t,x))+ \frac{\partial}{\partial x}\sum_j v^2_jq^j_{t}(x)=
-\frac{\partial V}{\partial x}\rho(t,x).
\end{align*}

Notice that the second term is a state derivative in
$x$ so the average over this - the integral over $x$
vanishes.  As a result now
\begin{align*}
\frac{\partial^2}{\partial t^2}E\left[x(t)\right]
&=\frac{\partial}{\partial t}\int_{-\infty}^{\infty}x\frac{\partial \rho(t,x)}{\partial t}dx=
\frac{\partial}{\partial t}\int_{-\infty}^{\infty}-x\frac{\partial v(t,x)\rho(t,x)}{\partial x}dx
\\
&=\frac{\partial}{\partial t}\int_{-\infty}^{\infty} v(t,x)\rho(t,x)dx
=-c\int_{-\infty}^{\infty}\left(\sum_{jk} v_j\omega_{jk}q^k(t,x)\right)dx
\\
&=\int_{-\infty}^{\infty}\left(\sum_{k} \frac{\partial
V}{\partial x}q^k(t,x)\right)dx
=\int_{-\infty}^{\infty}\frac{\partial V}{\partial x}\rho(t,x)dx=
E\left[\frac{\partial V}{\partial x}\right].
\end{align*}
\end{proof}

Notice that in this example there is no uniqueness for
the choices of $\alpha$ and $\beta$.  The actual form
of the transaction matrix is not clear and the actual
form in which these parameters depend on the potential
is surprising.

After this example let us take a look at the energy
embedded in equation \eqref{TRANSITION30} with rate
choices \eqref{TRANSITION31}.

\begin{thm}
Using the probability matrix defined in Theorem
\ref{POTENTIAL1} and using the definition of potential
in  \eqref{TRANSITION30} then
\begin{align*}
\frac{1}{2}E\left[\sum_j v_j^2 q^j(t,x)\right]=
\frac{1}{2}\int \sum_j v_j^2 q^j(t,x)dx
\end{align*}
is the average kinetic energy of the particle and
\begin{align*}
\frac{\partial}{\partial t}\left[\frac{1}{2}E\left[ \sum_j v_j^2 q^j(t,x)\right]dx+
E\left[V\right]\right]=
\frac{c^2}{2} (\alpha +\beta).
\end{align*}
\end{thm}
\begin{proof}
Using again \eqref{TRANSITION30} multiply each row
with $v_k^2=k^2c^2$ and sum over the equation to get
\begin{align*}
\frac{\partial}{\partial t} \sum_j v_j^2 q^j_{t}(x)+
\frac{\partial}{\partial x}\sum_j v^3_jq^j_{t}(x)=
\sum_{jk} v_j^2\omega_{jk}q^k(t,x).
\end{align*}

Take $v_k^2=k^2c^2$ then for all $k$ we have
\begin{align*}
\sum_{j} v_j^2\omega_{jk} = \sum_{j} c^2j^2\omega_{jk}
&=c^2 \left(\alpha (k-1)^2 -(\alpha + \beta) k^2 + \beta (k+1)^2\right)
\\
&=-2kc^2(\alpha -\beta) + c^2 (\alpha +\beta)
\\
&=-2ck\frac{\partial V}{\partial x}
+ c^2 (\alpha +\beta)
\end{align*}
so then
\begin{align*}
\frac{\partial}{\partial t} \sum_j v_j^2 q^j_{t}(x)+
\frac{\partial}{\partial x}\sum_j v^3_jq^j_{t}(x)&=
\sum_{k}\left( -2ck\frac{\partial V}{\partial x}
+ c^2 (\alpha +\beta)\right)q^k(t,x)
\\
&= -2c\frac{\partial V}{\partial x}\sum_{k}kq^k(t,x)
+ c^2 (\alpha +\beta)\rho(t,x)
\\
&= -2\frac{\partial V}{\partial x}v(t,x)\rho(t,x)
+ c^2 (\alpha +\beta)\rho(t,x).
\end{align*}

Now notice that
\begin{align*}
\frac{\partial}{\partial t} E\left[V\right]&=
\frac{\partial}{\partial t}\int V \rho(t,x)dx =
\int V \frac{\partial}{\partial t}\rho(t,x)dx
\\
&=-\int V \frac{\partial}{\partial x}\left(v(t,x)\rho(t,x)\right)dx =
\int\frac{\partial V}{\partial x}v(t,x)\rho(t,x)dx.
\end{align*}

Taking expectation on both sides noticing that the
partial $x$ term vanishes so that
\begin{align*}
\frac{1}{2}\int\frac{\partial}{\partial t} \sum_j v_j^2 q^j_{t}(x)dx+
\int\frac{\partial V}{\partial x}v(t,x)\rho(t,x)dx=
\frac{c^2}{2} (\alpha +\beta)\int\rho(t,x)dx
\end{align*}
or
\begin{align*}
\frac{\partial}{\partial t}\left[\frac{1}{2}E\left[ \sum_j v_j^2 q^j_{t}(x)\right]dx+
E\left[V\right]\right]=
\frac{c^2}{2} (\alpha +\beta).
\end{align*}
\end{proof}

So the energy defined in the system evaporates at a
rate $(\alpha +\beta)c^2/2$ depending on the choices
of the system as long the difference between $\alpha$
and $\beta$ is proportional to the differential of the
potential.  A logical choice for rate parameters
becomes
\begin{align*}
\alpha &=\theta + \frac{1}{2c}\frac{\partial V}{\partial x}
\\
\beta &=\theta - \frac{1}{2c}\frac{\partial V}{\partial x}
\end{align*}
where for the moment it is assumed that
\begin{align*}
\abs{\frac{1}{2c}\frac{\partial V}{\partial x}} << \theta.
\end{align*}
In this case $(\alpha +\beta)c^2/2=\theta c^2$.

Substituting this choice into \eqref{TRANSITION30} for
the $\omega$ matrix becomes
\begin{align*}
w &=\theta
\begin{pmatrix}
-2 & 1 &  &  &
\\
1 & -2 &  &  &
\\
 &  & \ddots &  &
\\
 &  &  & -2 & 1
\\
 &  &  & 1 & -2
\end{pmatrix}
+\frac{1}{2c}
\begin{pmatrix}
0 & \frac{\partial V}{\partial x} &  &  &
\\
-\frac{\partial V}{\partial x} & 0 & \frac{\partial V}{\partial x} &  &
\\
 &  & \ddots &  &
\\
 &  & -\frac{\partial V}{\partial x} & 0 & \frac{\partial V}{\partial x}
\\
 & &  & -\frac{\partial V}{\partial x} & 0
\end{pmatrix}
\\
&=\theta \widetilde{D} + \frac{1}{2c} \widetilde{\Gamma}
\end{align*}
with obvious definitions for the $NxN$ sized symmetric
matrix $\widetilde{D}$ and the equal-sized
antisymmetric matrix $\widetilde{\Gamma}$.

As a result
\begin{align*}
\frac{\partial}{\partial t} q^j(t,x)+v_j\frac{\partial}{\partial x} q^j(t,x)
=\sum_k \left(\theta \widetilde{D}_{jk}+\frac{1}{2c}\widetilde{\Gamma}_{jk}\right)q^k(t,x),
 \text{        $j=-N,..., -1,0,1,...,N$}
\end{align*}
for the symmetric and antisymmetric matrices above.
Notice that for this case the size of the equations
has been constrained to $2N+1$.

\subsection*{More concise form, speed}
It is possible to reduce the size of the original
equation by a certain amount though this may be
detrimental to the simplicity of the transition
matrix.

\begin{thm}
Let
\begin{align}
\label{TRANSITION48}
\psi^j(t,x)=\h_j(t)q^j(t,x)=e^{v_jt\frac{\partial}{\partial x}}q^j(t,x)
\end{align}
with $\h_j(t)$ is the translation operator
\begin{align}
\label{TRANSITION49}
\h_j(t)=e^{v_jt\frac{\partial}{\partial x}} \text{        $j=..., -1,0,1,...$}
\end{align}
for all velocities $v_j,j=...,-1,0,1,...$. Then
\begin{align*}
\frac{\partial}{\partial t}\psi^j(t,x)=\sum_k \left(\h_j(t)\omega_{jk}\h^{-1}_k(t)\right)\psi^k(t,x)
\end{align*}
where
\begin{align*}
\h_j(t)\omega_{jk}\h^{-1}_k(t)=e^{v_jt\frac{\partial}{\partial x}}\omega_{jk}e^{-v_kt\frac{\partial}{\partial x}}
\end{align*}
for all combinations of $j,k=...,-1,0,1,....$.
\end{thm}
\begin{proof}
Using equation \eqref{TRANSITION48} with operator
\eqref{TRANSITION49} the first part of
\eqref{TRANSITION30} can be written as
\begin{align*}
\frac{\partial}{\partial t} q^j(t,x)+v_j\frac{\partial}{\partial x} q^j(t,x)
=\h^{-1}_j(t)\frac{\partial}{\partial t}\h_j(t)q^j(t,x), \text{        $j=..., -1,0,1,...$}
\end{align*}
where $\h_j(t)$ is the translation operator
\begin{align*}
\h_j(t)=e^{v_jt\frac{\partial}{\partial x}} \text{        $j=..., -1,0,1,...$}
\end{align*}
for all velocities $v_j,j=...,-1,0,1,...$.

This operator translates the argument in a function
since for any test function $f=f(t,x)$
\begin{align*}
\h_j(t)f(t,x)=f(t+v_jt,t). \text{        $j=..., -1,0,1,...$}
\end{align*}
Also
\begin{align*}
\h_j(t)\h_k(t)&=\h_k(t)\h_j(t), \text{        $j=..., -1,0,1,...$}
\\
\h^{-1}_j(t)&=e^{-v_jt\frac{\partial}{\partial x}}. \text{        $j=..., -1,0,1,...$}
\\
\frac{\partial}{\partial t}\h_j(t)&=v_j\frac{\partial}{\partial x}\h_j(t)=v_j\h_j(t)\frac{\partial}{\partial x}. \text{        $j=..., -1,0,1,...$}
\end{align*}

Applying this to equation \eqref{TRANSITION30} yields
\begin{align*}
\frac{\partial}{\partial t}\h_j(t)q^j(t,x)=\sum_k \h_j(t)\omega_{jk}\h^{-1}_k(t)\h_k(t)q^k(t,x)
\end{align*}
which is equivalent to
\begin{align*}
\frac{\partial}{\partial t}\psi^j(t,x)=\sum_k \h_j(t)\omega_{jk}\h^{-1}_k(t)\psi^k(t,x)
\end{align*}
with
\begin{align*}
\psi^j(t,x)=\h_j(t)q^j(t,x)=e^{v_jt\frac{\partial}{\partial x}}q^j(t,x).
\end{align*}
\end{proof}

The other element is the energy flow through this
example.
\begin{thm}
In this case
\begin{align*}
a_t(x)&=\lim_{h\downarrow 0}\frac{v(t,x)-v^-(t,x)}{h}
=-c\epsilon(t,x)
\end{align*}
where $v(t,x)$ and $v^-(t,x)$ are the forward and
backward energies in the system.
\end{thm}
\begin{proof}
To find the equivalent of \eqref{TRANSITION26}
consider \eqref{TRANSITION21} in the small $h$ limit.
Invert \eqref{TRANSITION46} to get
\begin{align}
\label{TRANSITION27}
q^k(t-h,x-v_kh) =\sum_j \left(\delta_{jk}+h\omega_{jk}\right)^{-1}q^j(t,x)
\end{align}
for again $k=...,-1,0,1,...$.  Writing
\begin{align*}
\delta_{jk}+h\omega_{jk}=(I+h\omega)_{jk}
\end{align*}
the inverse can be written as
\begin{align*}
(I+h\omega)^{-1}(I+h\omega)=I
\end{align*}
and so for small $h$
\begin{align*}
(I+h\omega)^{-1}\approx I-h\omega.
\end{align*}

Now
\begin{align*}
v^-(t,x) = \sum_jv_jP[x(t-h)=x-j\Delta x | x_t=x]
\end{align*}
or
\begin{align*}
v^-(t,x) &= \sum_jv_j\frac{q^+(t-h,x-j\Delta x)}{\rho(t,x)}
\\
&=\sum_jv_j\sum_k\frac{\left(\delta_{jk}+h\omega_{jk}\right)^{-1}q^k(t,x)}{\rho(t,x)}
\\
&\approx\sum_jv_j\sum_k\frac{\left(\delta_{jk}-h\omega_{jk}\right)q^k(t,x)}{\rho(t,x)}
\\
&\approx\sum_jv_j\frac{\left(q^j(t,x)-h\sum_k\omega_{jk}\right)q^k(t,x)}{\rho(t,x)}
\\
&\approx\sum_jv_j\frac{q^j(t,x)}{\rho(t,x)}-h\frac{\sum_{jk}v_j\omega_{jk}q^k(t,x)}{\rho(t,x)}
\\
&=v(t,x) + hc\epsilon
\end{align*}
so acceleration through node state $x$ can be defined
as
\begin{align*}
a(t,x)=\frac{v(t,x)-v^-(t,x)}{h} = -c\epsilon.
\end{align*}
\end{proof}
So what is interesting is that in more dimensions the
inverse velocity is a simplification of the expression
derived in Theorem \ref{VELOCITY2}.

\section{Conclusion}
Section 1 showed the distribution for the position
(state) distribution of the single step binomial
process assuming that the velocity on the node grid
rather than the state process is Markovian.  The
result is a set of joint velocity distributions
related to a rate matrix.

The final distribution may show the original velocity
information and for sufficient small rates and
relatively small initial densities transport the
original conditions into the final distribution. As
the numerical example shows if the initial
distribution widen and the rates increase the
distribution focusses around the probabilistic drift
in a single modal distribution.  On the other hand for
extreme small rates and a very focussed initial
distribution the final density shows variations.

For a much smaller grid and constant rates the
probability equations converge into a correlated set
of probabilities of hyperbolic functions for each
velocity in state point.  The two dimensional case can
be transformed into a Telegraph equation for the state
density which can be transformed into a Klein-Gordon
equation if the transition rates are constant.  An
average velocity from the state can be defined as well
as Section III and Section IV show.

This equation can be done in two velocity spaces or in
an infinite number is a set of diffusion equations of
them.  Both for the two-dimensional applications and
the multi-dimensional case a forward and a backward
velocity can be found as Section III and Section V
show. In the last Section there is multi-dimensional
hyperbolic partial differential equation whose average
satisfies Newton's equation.



\begin{thebibliography}{99}

\bibitem{ANNO1} Anno, P.D., \textsl{Klein Gordon
    Acoustics Theory}, Thesis, Coloradon School of
    Mines (1993).

\bibitem{BOGUNA1} Boguna, M., Porra`J.M., Masoliver J,
    \textsl{Generalization of the persistent random walk to
    dimensions greater than 1}, Physical Review E 58 (6) (Dec
    1998).

\bibitem{CODLING1} Codling, E.A., Plank, M.J.,
    Benhamou, S., \textsl{Random Walk in Biology},
    Journal of Royal Society Interface 5 (25)  813-834.

\bibitem{FEINBERG} Feinberg, G, \textsl{The
    Possibility of Faster Than Light}, Physical Review
    159, Volume 5 (1967).

\bibitem{FELLER} Feller, W, \textsl{An Introduction to
    Probability Theory and Its Applications}, John
    Wiley and Sons, Volume I and II (1966, 1971).

\bibitem{GOLDSTEIN1} Goldstein, S, \textsl{On
    diffusion by discontinuous movements and the
    telegraph equation}, Quarter Journal Mechanics
    Applied Mathematics 4 (129-156).

\bibitem{HARRINGTON} Harrington. R.F.,
    \textsl{Time-Harmonic
    Electromagnetic Fields}, McGraw-Hill (1961).

\bibitem{HOSSEINI1} Hosseini, M.M., Tauseef
    Mohyud-Din, S, Hosseini, S.M. and Heydari, M.,
    \textsl{Study on Hyperbolic Telegraph Equations by Using Homotopy Analysis Method},
    Studies in Nonlinear Sciences 1, (2) (2010) 50-56.

\bibitem{IBE1} Ibe, O.C. \textsl{Elements of Random
    Walk and Diffusion Processes}, John Wiley Series in
    Operational Research and Operational Science, John Wiley
    (2013).

\bibitem{IACUS1} Iacus, S.M., \textsl{Statistical
    analysis of the inhomogeneous telegrapher's
    process}, Cornell University Library arXiv.org,
    arxiv Engine, http:/arxiv.org/abs/math/0011059v1 (2000).

\bibitem{JAVIDI1} Javidi, M., Nyamoradi, N.,
    \textsl{Numerical solution of telegraph equation using LT inversion Technique},
    International Journal of Advanced Mathematical Sciences, 1 (2) (2013)
    64-77.

\bibitem{KAC1} Kac, M.,
    \textsl{A stochastic model related to the Telegraphers equation},
    Rocky Mountain Journal Mathematics 4 (1974) 497-509.
    Reprinted from: M. Kac, \textsl{Some stochastic problems in
    physics and mathematics}, Colloquium lectures in the
    pure and applied sciences, No. 2, hectographed, Field
    Research Laboratory, Socony Mobil Oil Company, Dallas,
    TX (1956) 102-122.

\bibitem{MITTAL1} Mittal, R.C., Bhatia, R.
    \textsl{Numerical solution to second order one dimensional Telegraph Equation
    by cubic B-spline collocation method}, Applied
    Mathematics and Computation 220, (2013) 496-506.

\bibitem{RATANOV1} Ratanov, N.,
    \textsl{Telegraph processes with jump diffusion and complete market models},
    Cornell University Library arXiv.org,
    http://arxiv.org/pdf/1311.5464.pdf (2013).

\bibitem{SAMOILENKO1} Samoilenko, I.V., Turbin, A.F.
    \textsl{A probability method for the solution of the telegraph equation
    with real-analytic initial conditions},
    Ukrainian Mathematical Journal 52 (8) (2000).

\end{thebibliography}
\end{document}